\documentclass[12pt,a4paper]{amsart}
\usepackage[cmtip,arrow]{xy}
\usepackage{amsxtra}
\usepackage{enumitem, fancyref, theoremref}
\usepackage{amsmath, amssymb, tikz, amsthm, thmtools}
\usepackage{hyperref, caption}

\calclayout

\DeclareMathSymbol{\subsetneq}{\mathrel}{AMSb}{"28}
\DeclareMathSymbol{\rightrightarrows}{\mathrel}{AMSa}{"13}

\newcommand{\BZ}{\mathbb{Z}} %integer numbers
\newcommand{\LSeq}{\sim_c} %LS equivalence
\newcommand{\derG}[1]{\stackrel{#1}{\Rightarrow}} %derive
\newcommand{\G}{{\mathcal G}}

\theoremstyle{plain}

\newtheorem{theorem}{Theorem}[section]
\newtheorem{lemma}[theorem]{Lemma}

\theoremstyle{remark}

\theoremstyle{definition}

\begin{document}

\title{Nonterminal complexity of some families of infinite regular languages}
\author{Dmitry Golubenko}

\address{Faculty of Mathematics, Higher School of Economics, 6~Usacheva str., Moscow, Russia}
\email{golubenko@mccme.ru}

%\address{Universit'e Claude Bernard, Lyon 1, IITP RAS, and HSE university}
%\email{kroshnin@math.univ-lyon1.fr}

\begin{abstract}
Nonterminal complexity of a context-free language is the smallest possible number of nonterminals in its generating grammar. While in general case nonterminal complexity computation problem is unsolvable, it can be computed for different families of regular languages. In this paper we study nonterminal complexity of some families of infinite regular languages. 
\end{abstract}

\maketitle

\tableofcontents

\section{Introduction}

For a context-free language $L$ we define its nonterminal complexity $Var(L)$ as the smallest posible number of nonterminals in a context-free grammar generating $L$:
\begin{align*}
Var(L) = \min \left\{ |N| \colon L \text{ is generated by } (N, \Sigma, P, S) \right\}
\end{align*}
It seems that Gruska was first to study nonerminal complexity. In \cite{grushka} he proved that for every $n$ there exists regular language $L_n$ over alphabet $\{a, b\}$ such that $Var(L_n) = n$, namely for $n \geqslant 2$
\begin{align*}
Var\left((ab)^* + \ldots + (ab^n)^*\right) = n + 1
\end{align*}
Gruska also proved that over one-letter alphabet every context-free language has nonterminal complexity at most 2 and
\begin{align*}
Var(a^2 + (a^3)^*) = 2
\end{align*}
In \cite{dassow} Dassow and Stiebe studied behavior of nonterminal complexity w. r. t. language operations such as union, conctatenation, Kleene star and homomorphisms. For example, if $Var(L_1) = n_1$ and $Var(L_2) = n_2$, then one easily constructs context-free grammar with $n_1 + n_2 + 1$ nonterminals which generates $L_1 \cup L_2$; thus $Var(L_1 \cup L_2) \leqslant n_1 + n_2 + 1$. Dassow and Stiebe show that this naive estimation is exact by constructing languages $L_{m,n}^{(k)}$ and $K_{m,n}^{(k)}$ such that
\begin{align*}
Var(L_{m,n}^{(k)}) = m, \quad Var(K_{m,n}^{(k)}) = n, \quad Var(L_{m,n}^{(k)} \cup K_{m,n}^{(k)}) = k
\end{align*}
for any $k \leqslant m + n + 1$. The key arguement is the following statement.
\begin{lemma}[\cite{dassow}]
	For pairwise different integers $k_1, k_2 \ldots k_{2n}$ we have
	\begin{align*}
	Var\left( (ab^{k_1})^* (ab^{k_2})^* \ldots (ab^{k_{2n}})^* \right) = n
	\end{align*}
\end{lemma}
Dassow and Stiebe also prove similar results for $Var(L_1 \cup L_2)$, $Var(L_1 L_2)$, $Var(L_1^*)$ and $Var(h(L_1))$. 

Nonterminal complexity of finite languages was developed extensively because of its connections with grammar compression and proof theory. Indeed, in \cite{hetzl} it was shown that cut-eliminaion process for a certain class of proofs corresponds to computation of the language of a tree grammar. Thus compression of finite language with the smallest possible grammar is used to produce shorter proofs. See \cite{wolf} for details.

In this paper we study nonterminal complexity of regular languages. Starting from a toy example discussed in section 2, we estimate $Var(w_1^* + \ldots + w_n^*)$ for different words $w_1, \ldots w_n$.

\section{Toy example}

This example was first given as a question for university course final exam.

\begin{lemma}
	There exist no CFG with one nonterminal which generates $L = (ab)^* + (ba)^*$.
\end{lemma}

\begin{proof}
	Suppose that there exists CFG $\G = (\{S\}, \{a, b\}, P, S)$ such that $L = L(\G)$. Every production is $S \rightarrow \alpha$ for some $\alpha \in \{a, b, S\}^*$, say, $\alpha = w_1 S w_2 S \ldots w_k S w_{k+1}$. Then $$u_1, \ldots u_k \in L \Rightarrow w_1 u_1 w_2 u_2 \ldots w_k u_k w_{k+1} \in L,$$ because as $L$ is generated by $\G$, one can derive this word by using $S \rightarrow w_1 S w_2 S \ldots w_k S w_{k+1}$ and then deriving every $u_i$ from nonterminals. Them $S$ can't be followed immediately by any symbol in right-hand side of every production:
	
	\begin{itemize}
		\item if $S$ is followed by $a$ (or $b$), we can derive $ba$ (or $ab$) from $S$, thus we can derive the word which contains two $a$'s (or two $b$'s) in the row, which doesn't lie in $L$;
		\item if $S$ follows $a$ (or $b$), we can derive $ab$ (or $ba$) from $S$, thus we can derive the word which contains two $a$'s (or two $b$'s) in the row, which doesn't lie in $L$;
		\item if we have subword $SS$ of $\alpha$, we can derive $ab$ from the first $S$ and $ba$ from the second.
	\end{itemize}
	
	So $\G$ can only have productions of sort $S \rightarrow S$ and $S \rightarrow w$ for $w \in L$. Thus $L(\G)$ is finite, which is contradiction.
\end{proof}

However, there exists CFG with two variables which generates $L$. Indeed, we may just consider the following CFG $$ S \rightarrow bAa \mid A , \, A \rightarrow abA \mid \varepsilon.$$

\section{Languages of kind $w_1^* + \ldots + w_n^*$}

Now let's consider language of kind $w_1^* + \ldots + w_n^*$ where $w_1, \ldots w_n$ are different nonempty words over an arbitrary alphabet.

For every nonempty word $u$ we define $\widehat{u}$ as a primitive word such that $u = \widehat{u}^k$ for some natural $k$. This word $\widehat{u}$ exists and is unique (see \cite{shallit}).

Thus we may introduce equivalence relation: $x \LSeq y$ iff $\widehat{x} = zw$ and $\widehat{y} = wz$ for some words $w$ and $z$. It's easy to check that $\LSeq$ is indeed equivalence relation since words $zw$ and $wz$ are neither or both primitive and $x \LSeq y$ iff $|\widehat{x}| = |\widehat{y}| = m$ and $\widehat{y}[i] = \widehat{x}[(k+i) \mod m]$ for every $i$ and some $k$ from $\{0, 1, \ldots m-1\}$.

We'll also suppose by default (without loss of generality) that there is no rules of kind $A \rightarrow A$, no empty nonterminal (such that there is no word derived from it) and no useless nonterminals (i. e. every nonterminal appears in derivation of some word).

For every CFG $\G = (N, \Sigma, P, S)$ one can construct a digraph $D(\G) = (N, E)$ where $N$ is the set of vertices and 
\begin{align*}
E = \{ (A, B) \colon A \rightarrow \alpha B \beta \in P  \text{ for some } \alpha, \beta \in (N \cup \Sigma)^* \}
\end{align*}

%We'll write $A \creat B$ if there exist a path from $A$ to $B$ in $D(\G)$.

\begin{lemma}
	\label{Scycle}
	Suppose that $\G$ is CFG such that $L(\G) = w_1^* + \ldots + w_n^*$ and there exists nontrivial path from $S$ to $S$ in $D(\G)$. Then words $w_1, \ldots w_n$ commute pairwise.
\end{lemma}

\begin{proof}
	Since $w_1^* + \ldots + w_n^*$ is infinite and $D(\G)$ contains nontrivial path from $S$ to $S$, there exist strings $\alpha, \beta \in (N \cup \Sigma)^*$ such that $S \derG{*} \alpha S \beta$ and $\alpha \beta \neq \epsilon$. Then by deriving some word from each nonterminal in $\alpha$ and $\beta$ we obtain $S \derG{*} x S y$ and $xy \neq \epsilon$.
	
	Let $M = \prod_{i = 1}^{n} |w_i|$. For any $i \in \{1, \ldots n\}$ there exists $j$ such that
	\begin{align*}
	S \derG{*} x^{M} S y^{M} \derG{*} x^{M} w_i^{\frac{M}{|w_i|}} y^{M} \in w_j^*
	\end{align*}
	As soon as $|x^M|$ and $|y^M|$ are divided by $|w_j|$ we hav $x^M = w_j^{b_x}$, $y^M = w_j^{b_y}$ and $w_i^{\frac{M}{|w_i|}} = w_j^{b_w}$ for some integers $b_x, b_y, b_w$. Thus $x, y, w_i \in p^*$ for some primitive word $p$. As $w_1, \ldots w_n \in p^*$ it follows that $w_1, \ldots w_n$ commute pairwise.
\end{proof}

\begin{theorem}
	The language $w_1^* + \ldots + w_n^*$ can be generated by CFG with one variable iff words $w_1, \ldots w_n$ commute pairwise.
\end{theorem}

\begin{proof}
	If $w_1^* + \ldots + w_n^* = L(\G)$ for CFG $\G = (\{S\}, \Sigma, P, S)$ then $\G$ satisfies conditions of lemma \ref{Scycle} and thus $w_1, \ldots w_n$ commute pairwise. 
	
	If $w_1, \ldots w_n$ commute pairwise then according to Lyndon-Sch\"{u}tzenberger there exists some word $w$ and integers $k_1, \ldots k_n$ such that $w_i = w^{k_i}$ for every $i \in \{1, \ldots n\}$. Let
	\begin{align*}
	R = \{ r \in \BZ_{\geqslant 0} \colon \exists u \in w_1^* + \ldots + w_n^* \quad |u| = r ~\mathrm{mod}~ \prod_{i = 1}^n k_i \}
	\end{align*}
	Then there exists grammar
	\begin{align*}
	S \rightarrow w^{\prod_{i = 1}^n k_i} S | w^{r_1} | \ldots | w^{r_m}, \text{  where  } R = \{r_1, \ldots r_m\}
	\end{align*}
	It's easy to check that this grammar generates $w_1^* + \ldots + w_n^*$.  
\end{proof}

Suppose now that there is no nontrivial path from $S$ to $S$ in $D(\G)$.

\begin{theorem}
	Let $k > 1$. Then $Var(w_1^* + \ldots + w_n^*) = k$ iff set $\{w_1, \ldots w_n\}$ is divided into $k-1$ $\LSeq$-equivalence classes.
\end{theorem}

\begin{proof}
	Suppose that words $w_1, \ldots w_n$ are pairewise $\LSeq$-equivalent. We show how to construct CFG with two nonterminals generating $w_1^* + \ldots + w_n^*$. As $w_i \LSeq w_j$ for every $i$ and $j$ there exists primitive word $w$ such that $\widehat{w_i} \LSeq w$ for every $i$. Each word $w_i$ can be represented as $(s_i p_i)^{k_i}$ where $p_i s_i = w$. Let
	\begin{align*}
	R = \{ r \in \BZ_{\geqslant 0} \colon \exists u \in w_1^* + \ldots + w_n^* \quad |u| = r ~\mathrm{mod}~ \prod_{i = 1}^n k_i \}
	\end{align*}
	Then there exists grammar
	\begin{align*}
	S \rightarrow s_1 A p_1 | \ldots | s_n A p_n, \quad
	A \rightarrow w^{\prod_{i = 1}^n k_i} A | w^{r_1-1} | \ldots | w^{r_m-1}, \text{  where  } R = \{r_1, \ldots r_m\}
	\end{align*}
	Here $r_j - 1$ are considered as residues modulo $M$. It's easy to check that this grammar generates $w_1^* + \ldots + w_n^*$.
	
	If set $\{w_1, \ldots w_n\}$ is divided into $k-1$ $\LSeq$-equivalence classes we may use the construction above to find a CFG with $k$ nonterminals generating $w_1^* + \ldots + w_n^*$.
	
	Suppose now that $Var(w_1^* + \ldots + w_n^*) = k$. First we prove that in every grammar $\G$ with $k$ nonterminals generating language $w_1^* + \ldots + w_n^*$ each grammar $(N, \Sigma, P, A)$ generates infinite language. Suppose that there exists nonterminal $A$ such that $(N, \Sigma, P, A)$ is finite, then there exists some $B \in N$ such that if $B \rightarrow \gamma \in P$ then $\gamma \in \Sigma^*$ --- it may be $A$ or some nonterminal reachable from $A$ in $D(\G)$. Thus nonterminal $B$ may be omitted and its every occurence in right-hand sides of rules may be interchanged with every $\gamma$ such that $B \rightarrow \gamma \in P$. Indeed, if $B \rightarrow \gamma_1|\ldots|\gamma_m$ we may replace every rule of kind $A \rightarrow \beta_1 B \beta_2 B \ldots B \beta_{k+1}$ with the collection of rules of kind $A \rightarrow \beta_1 \gamma_{i_1} \beta_2 \gamma_{i_2} \ldots \gamma_{i_k} \beta_{k+1}$. Thus we obtain CFG with less than $k$ nonterminals generating $w_1^* + \ldots + w_n^*$ which is contradiction. 
	
	%Suppose now that there exist two different nonterminals $A, B \neq S$ such that there is a path from $A$ to $B$ in $D(\G)$. Then ...
	
	Let $A \neq S$ be a nonterminal. As soon as language $(N, \Sigma, P, A)$ is infinite, there exists such word $u \in \Sigma^*$ that $A \derG{*} u$ and $|u| \geqslant 2^n M$ where $M = \prod_{i = 1}^{n} |w_i|$. Suppose that
	\begin{align*}
	S \derG{*} \alpha_i A \beta_i \derG{*} x_i u y_i = w_i^{l_i}, \quad S \derG{*} \alpha_j A \beta_j \derG{*} x_j u y_j = w_j^{l_j}
	\end{align*}
	then there exist such words $z_i$ and $t_i$ that $x_i u y_i = x_i z_i w_i^{\frac{2M}{|w_i|}} t_i y_i$ and $|x_i z_i|$ is divided by $M$; then for some words $p_j$ and $s_j$ such that $w_i^{\frac{M}{|w_i|}} = p_j s_j$ and $|x_j z_i p_j|$ is divided by $M$ we have
	\begin{align*}
	x_j u y_j = x_j z_i w_i^{\frac{2M}{|w_i|}} t_i y_j = x_j z_i (p_j s_j)^2 t_i y_j = x_j z_i p_j (s_j p_j) s_j t_i y_j
	\end{align*}
	Thus $|s_j p_j| = M$ so $s_j p_j = w_j^{\frac{2M}{|w_j|}}$. This yields $w_i \LSeq w_j$.
	
	Thus $\{w_1, \ldots w_n\}$ is divided into at most $k-1$ $\LSeq$-equivalence classes. If the number of classes is less than $k-1$ then one can construct CFG with less than $k$ nonterminals which generates $w_1^* + \ldots + w_n^*$.
\end{proof}

\section{Acknowledgements}
The author would like to thank Alexey Kroshnin, Meruza K. and Igor Shimanogov for hospitality and useful discussions.
%\bigskip\addtocontents{toc}{\medskip}

\hbadness=1100

\end{document}